\documentclass[english,a4paper, 10pt]{extarticle}

\usepackage[a4paper, total={7in, 10in}]{geometry}
\usepackage{amsmath}
\usepackage{amssymb}
\usepackage{amsfonts}
\usepackage{bm}
\usepackage{braket}
\usepackage{qcircuit}
\usepackage[]{multicol}
\usepackage{latexsym}
\usepackage{mathtools}
\usepackage{here}
\usepackage[dvipdfmx]{graphicx}
\usepackage[
    backend = biber,
    bibstyle = nature,
    sorting = none,
]{biblatex}
\usepackage[]{hyperref}
\usepackage{amsthm}
\usepackage{comment}
\usepackage{color}

\makeatletter

\newenvironment{figurehere}
  {\def\@captype{figure}}
  {}
\makeatother

\theoremstyle{plain}
\newtheorem{theorem}{Theorem}
\newtheorem{lemma}{Lemma}
\mathtoolsset{showonlyrefs=true}

\title{Graph kernels encoding features of all subgraphs \\ by quantum superposition}

\author{Kaito Kishi\thanks{Department of Biosciences \& Informatics, Keio University, Hiyoshi 3-14-1, Kohoku, Yokohama 223-8522, Japan} \and Takahiko Satoh\thanks{Quantum Computing Center, Keio University, Hiyoshi 3-14-1, Kohoku, Yokohama 223-8522, Japan} \and Rudy Raymond\footnotemark[2] \thanks{IBM Quantum, IBM Research - Tokyo, 19-21, Nihonbashi Hakozaki-cho Chuo-ku, Tokyo 103-8510, Japan} \and Naoki Yamamoto\footnotemark[2] \thanks{Department of Applied Physics and Physico-Informatics, Keio University, Hiyoshi 3-14-1, Kohoku, Yokohama 223-8522, Japan} \and Yasubumi Sakakibara\footnotemark[1]}

\begin{document}

\maketitle

\begin{abstract}
Graph kernels are often used in bioinformatics and network applications to measure the similarity 
between graphs; therefore, they may be used to construct efficient graph classifiers. 
Many graph kernels have been developed thus far, but to the best of our knowledge there is no 
existing graph kernel that considers all subgraphs to measure similarity. 
We propose a novel graph kernel that applies a quantum computer to measure the graph similarity taking 
all subgraphs into account by fully exploiting the power of quantum superposition to encode every subgraph 
into a feature.
For the construction of the quantum kernel, we develop an efficient protocol that removes the index 
information of subgraphs encoded in the quantum state. 
We also prove that the quantum computer requires less query complexity to construct the feature 
vector than the classical sampler used to approximate the same vector. 
A detailed numerical simulation of a bioinformatics problem is presented to demonstrate that, in 
many cases, the proposed quantum kernel achieves better classification accuracy than existing 
graph kernels.
\end{abstract}

\begin{multicols}{2}


\section*{Introduction}

An effective measure of the similarity between graphs is necessary in several science and 
engineering fields, such as bioinformatics, chemistry, and social networking \cite{Vishwanathan}. 
In machine learning, this measure is called the graph kernel, and it can be used to construct 
a classifier for graph data \cite{BorgwardtProteinPred}. 
In particular, a kernel in which all subgraphs are fully encoded is desirable, because it can access 
the complete structural information of the graph. 
However, constructing such a kernel is known as a nondeterministic polynomial time (NP)-hard 
problem \cite{Gartner}. 
Alternatively, a kernel that encodes partial features of {\it all} subgraphs may be used, but to 
our best knowledge, this type of method has not been developed thus far. 
Previous studies have instead focused on using different features originating from the target 
graphs, such as random walks \cite{Gartner, Kashima, FastVishwanathan}, graphlet sampling 
\cite{N-Shervashidze}, and shortest paths \cite{Borgwardt}.

A quantum computer may be applied to construct a graph kernel that covers all subgraphs, because 
of its strong expressive power, which has been demonstrated in the quantum machine learning 
scenario \cite{QMLReview,Ciliberto_2018}. 
More specifically, the exponentially large Hilbert space of quantum states may serve as an appropriate 
feature space where the kernel is induced \cite{QMLinFeatureHilbert, QSVM,schuld2021quantum}.
This kernel can then be further utilized for machine learning.
We now have two approaches: the implicit (hybrid) approach, in which the quantum computer computes 
the kernel, and the classical computer uses it for machine learning, and the explicit approach, in which 
both the kernel computation and the machine learning part are both executed by the quantum computer.

In fact, there exist a few proposals for the quantum computational approach to construct graph kernels. 
For example, Ref. \cite{GBS-kernel} proposed the Gaussian Boson sampler, which estimates feature 
vectors by sampling the number of perfect matchings in the set of subgraphs. 
Another method used a quantum walk \cite{QWalkCont, QWalkDisc}. 
However, there is no existing graph kernel that operates on a quantum circuit to design the features 
obtained from all subgraphs.

In this study, we propose a quantum computing method to generate a graph kernel that 
extracts important features from all subgraphs. 
The point of this method is that the features of an exponential number of subgraphs can be effectively 
embedded to a quantum state in a Hilbert space using quantum computing. 
Note that a naive procedure immediately induces a difficulty; the corresponding quantum state 
contains the index component, which may severely decrease the value of the kernel. 
It is generally difficult to \textit{forget} the index component, as argued in~\cite{Aharonov-SZKcomp}, 
which we simply refer to as \textit{removing} the index component; nonetheless, we propose a protocol 
to achieve this goal using a polynomial number of operations (i.e., query complexity) under a valid 
condition, which is fortunately satisfied by the features used in typical problems in bioinformatics. 
We then provide some concrete protocols to further compute the target kernel and discuss their 
query complexity. 
Also, we prove that they require fewer operations to generate the feature vector than a classical 
sampler used to approximate the same vector. 
Hence, up to the difference of the sense of complexities, the proposed protocols have quantum 
advantage. 
Lastly, we use the above-mentioned typical bioinformatics problem to investigate whether the 
classifier based on the proposed quantum kernel achieves a higher classification accuracy than 
existing graph classifiers.


\section*{Results}

\subsection*{Algorithm for graph kernel computation} 

We consider a graph characterized by the pair \(G=(V,E)\), where \(V=\{v_1,v_2,\cdots,v_n\}\) is 
an ordered set of \(n\) vertices, and \(E\subseteq V\times V\) is a set of undirected edges. 
Hereafter, we use the notation \(n=|V|\). 
Also, let \(d\) be the maximum degree of the graph \(G\). 
In this study, \(G\) is assumed to be a simple undirected graph that does not contain self-loops 
or multiple edges.

The first step of our algorithm is to encode the graph information of $G$ onto a quantum state 
defined on the composite Hilbert space ${\mathcal H}_{\rm index}\otimes {\mathcal H}_{\rm feature}$. 
The index space ${\mathcal H}_{\rm index}$ is composed of $n$ qubits, which identifies a subgraph 
characterized by a set of vertices represented by the binary sequence $x$ of length \(n\). 
The feature space ${\mathcal H}_{\rm feature}$ is composed of $m$ qubits, each state of which 
represents the feature information of a chosen subgraph $x$. 
The value of $m$ depends on what feature is used.
In this study, we consider the case $m=O(\log n)$, where each qubit represents the numbers of 
vertices, edges, and vertices with a degree of 1, 2, and 3; 
refer to Toy Example section for a concrete example.

Now, we assume an oracle operator that encodes the feature information of the chosen subgraph, 
identified by the index $x$, to the function $E(G,x)$ and then generates the feature state 
$\ket{E(G,x)}$. 
Then, using the superposition principle of quantum mechanics, we can generate the quantum 
state containing the features of {\it all} subgraphs of a graph \(G\) as follows: 
\begin{align} 
\label{eq:g_with_index}
    \ket{\bar{G}}:=\sum_{x\in\{0,1\}^n} \ket{x}\ket{E(G,x)}, 
\end{align}
where the normalized coefficient is omitted to simplify the notation.

Next, we aim to compute the similarity of two graphs $G$ and $G'$. 
For this purpose, it seems that the inner product of $\ket{\bar{G}}$ and $\ket{\bar{G}'}$ may be used. 
Note that when the two graphs have the same feature $E(G,x_1) = E(G',x_2)$ with different 
indices $x_1\neq x_2$, they should still contribute to the similarity of $G$ and $G'$, whereas the inner 
product of $\ket{x_1}\ket{E(G,x_1)}$ and $\ket{x_2}\ket{E(G',x_2)}$ is zero. 
Therefore, what we require is the state 
\begin{align} 
\label{eq:g_without_index}
    \ket{G}:=\sum_{x\in\{0,1\}^n}\ket{E(G,x)}, 
\end{align}
instead of Eq. \eqref{eq:g_with_index}. 
However, an exponential number of operations is generally necessary to remove 
the index state \cite{Aharonov-SZKcomp}. 
The first contribution of our study is that our algorithm only needs a polynomial number of operations to obtain 
Eq.~\eqref{eq:g_without_index} from Eq.~\eqref{eq:g_with_index} under a condition that 
may be satisfied in features useful for many graph classification problems.

\begin{figurehere}
    \begin{align}
        \Qcircuit @C=1em @R=.7em {
            \lstick{\ket{0^n}} & {/} \qw & \multigate{1}{\textrm{Oracle}} & \gate{H^{\otimes n}} & \meter & \qw \\
            \lstick{\ket{0..}} & {/} \qw & \ghost{\textrm{Oracle}}        & \qw                  & \qw    & \qw
        }
    \end{align}
    \caption{Quantum circuit used to prepare the state \eqref{eq:g_without_index}.} 
    \label{qc:init-index-reg}
\end{figurehere} \noindent
\\

The algorithm to remove the index state is described in terms of the general discrete 
function \(f\) as follows (see also Fig.~\ref{qc:init-index-reg}): 
The state generated via the oracle operation is rewritten as 
\begin{equation}
      \frac{1}{\sqrt{2^n}}\sum_{x\in\{0,1\}^n} \ket{x}\ket{f(x)} 
          = \frac{1}{\sqrt{2^n}}\sum_{k=1}^a \sum_{x\in X_k} \ket{x}\ket{y_k}, 
\end{equation}
where \(X_k=\{x \, |\, f(x)=y_k,x\in\{0,1\}^n\}\). 
Also, \(a=|Y|\), where \(Y\) is the range of \(f\), i.e., \(Y=\{y_k\}_{k=1}^a\). 
Then, we apply \(H^{\otimes n} \otimes I \) to obtain 
\begin{equation} 
\label{eq:has-other-terms}
        \frac{1}{2^n}\ket{0^n}\left(\sum_{k=1}^a |X_k|\ket{y_k}\right)+\mathrm{other\ terms}, 
\end{equation}
where the ``other terms" contain all quantum states with index states other than $\ket{0^n}$. 
We now make a measurement in the computational basis of the index state; 
then, the post-selected feature state obtained when the measurement result is \(0^n\) is given by 
\begin{align}
\label{eq:thm1-final-state}
     \frac{1}{\sqrt{\sum_{k=1}^a |X_k|^2}} \sum_{x\in\{0,1\}^n} \ket{f(x)}, 
\end{align}
which is our target state. 
The probability to successfully obtain this state is 
\begin{equation} 
\label{eq:prob-to-init-ir}
        \Pr(0^n)=\frac{1}{2^{2n}}\sum_{k=1}^a |X_k|^2.
\end{equation}
To evaluate the efficiency of the proposed algorithm, we derive the minimum of the probability 
\eqref{eq:prob-to-init-ir} with respect to the family of set \(X=\{X_1,\cdots,X_k\}\) under the 
condition $\sum_{k=1}^a |X_k|-2^n=0$. 
Hence the problem is to maximize 
\begin{equation}
        L(X,\lambda)=-\sum_{k=1}^a |X_k|^2 - \lambda\left(\sum_{k=1}^a |X_k|-2^n\right), 
\end{equation}
where $\lambda$ is the Lagrange multiplier. 
Then, \(|X_k|=2^n/a\) maximizes \(L(X,\lambda)\); thus, Eq. \eqref{eq:prob-to-init-ir} has the 
following lower bound: 
\begin{equation} 
\label{eq:prob-lower-bound}
        \Pr(0^n)=\frac{1}{2^{2n}}\sum_{k=1}^a |X_k|^2 
              \geq \frac{1}{2^{2n}}a\left(\frac{2^n}{a}\right)^2=a^{-1}.
\end{equation}
The above result is summarized as follows:

\begin{theorem} \label{thm:rir}
The quantum circuit depicted in Fig.~\ref{qc:init-index-reg} generates the quantum state 
\eqref{eq:thm1-final-state} with a probability of at least \(a^{-1}\). 
\end{theorem}

Therefore, the quantum state \eqref{eq:thm1-final-state} can be effectively generated if $a=|Y|$ is 
of the order of a polynomial with respect to $\{y_k\}$. This desirable assumption holds 
in the encoding function $E(G,x)$ that is used in the bioinformatics problem analyzed later in this paper. 
For a general statement of this fact, we introduce the following constraint on the range of 
the function: 
the set 
\begin{align}
    X_v=\{x \, | \, f(x)\in Y_v,x\in\{0,1\}^n\}, ~ n\in\mathbb{N}
\end{align}
and \(Y_v\) specified by 
\begin{align}
    \bigcup_{v=0}^n Y_v=Y, ~~
    \bigcup_{v=0}^n \bigcup_{\substack{v'\neq v,\\v'\in[0,n]}}\left(Y_v\cap Y_{v'}\right)=\emptyset, 
\end{align}
are assumed to satisfy 
\begin{align} \label{eq:v_range_constraint}
    |X_v|=\binom{n}{v}, ~ |Y_v|=\begin{cases}
        1 & (v=0) \\
        {\rm Pol}(v^{c-1}) & ({\rm otherwise})
    \end{cases},
\end{align}
where ${\rm Pol}(v^{c-1})$ is a polynomial function with maximum degree $c-1$.  
These conditions lead to 
\begin{equation}
\label{eq:order of |Y|}
        a=|Y|=\sum_{v=0}^n {\rm Pol}(v^{c-1})= O(n^c).
\end{equation}
Then, Theorem \ref{thm:rir} can be further refined as follows 
(the proof is given in the Supplementary Information):

\begin{theorem} \label{thm:rir-v-range}
Given the condition \eqref{eq:v_range_constraint}, the algorithm depicted in Fig.~\ref{qc:init-index-reg} 
generates the state \eqref{eq:thm1-final-state} with a probability of at least \(\Omega(\sqrt{n}/n^c)\). 
\end{theorem}

In what follows we assume that the encoding function $E(G, x)$ satisfies the condition 
\eqref{eq:v_range_constraint}; then the desired state transformation \eqref{eq:g_with_index} 
$\to$ \eqref{eq:g_without_index} only requires a \(O(n^c/\sqrt{n})\) mean query complexity. 
As a consequence, we are now able to effectively compute the inner product 
\begin{align} 
\label{eq:kernel_def}
    \braket{G|G'}, 
\end{align}
as the similarity measure of the two graphs $G$ and $G'$ 
(note that both $\ket{G}$ and $\ket{G'}$ are real vectors). 
The task of computing \eqref{eq:kernel_def} is typically done via the swap test \cite{swaptest}. 
A diagram of the post-selection process from $\ket{0}\ket{\bar{G}}\ket{\bar{G}'}$ 
to $\ket{0}\ket{G}\ket{G'}$ is depicted in Fig. \ref{qc:swap-test}. 
The total mean query complexity required to prepare both $\ket{G}$ and $\ket{G'}$ and 
subsequently apply the swap test to compute the inner product \eqref{eq:kernel_def} is given by 
\begin{equation}
\label{simple complexity for inner}
    O(n^c/\sqrt{n}) + O(n'\mbox{}^c/\sqrt{n}) = O(n^c/\sqrt{n}), 
\end{equation}
where $n$ and $n'$ are assumed to be of the same order. 
Note that Eq.~\eqref{simple complexity for inner} contains the constant overhead required to repeat 
the swap test circuit to compute the inner product with a fixed approximation error. 
The resulting inner product computed by the swap test is represented by 
\begin{align}
\label{eq:kernel_def_nff}
    K_{\rm BH}(G, G')=k_{\rm BH}(G, G') f_G^\top f_{G'},
\end{align}
where $f_G=[|X_1|, \ldots, |X_a|]^\top$ is the column vector and the coefficient $k_{\rm BH}(G, G')$ 
is given by 
\begin{align} 
\label{eq:swap_test_coeff}
    k_{\rm BH}(G, G')=\frac{1}{\sqrt{\sum_{k=1}^a |X_k|^2}\sqrt{\sum_{k=1}^a |X_k'|^2}}.
\end{align}
Importantly, Eq.~\eqref{eq:kernel_def_nff} is exactly the Bhattacharyya (BH) kernel 
\cite{Bhattacharyya1943, BhattacharyyaKernel}, which has been successfully applied to image 
recognition \cite{BhattacharyyaKernel} and text classification \cite{BhattacharyyaKernelText}.

\begin{figurehere}
    \begin{align}
        \Qcircuit @C=1em @R=.7em {
            \lstick{\ket{0}}      & \qw     & \qw               & \qw    & \gate{H} & \ctrl{4} & \gate{H} & \meter \\
            \lstick{\ket{0^n}}    & {/} \qw & \multigate{1}{U}  & \meter & \qw      & \qw      & \qw      & \qw    \\
            \lstick{\ket{0..}}    & {/} \qw & \ghost{U}         & \qw    & \qw      & \qswap   & \qw      & \qw    \\
            \lstick{\ket{0^{n'}}} & {/} \qw & \multigate{1}{U'} & \meter & \qw      & \qw      & \qw      & \qw    \\
            \lstick{\ket{0..}}    & {/} \qw & \ghost{U'}        & \qw    & \qw      & \qswap   & \qw      & \qw
        }
    \end{align}
    \caption{
    Quantum circuit to compute the inner product \eqref{eq:kernel_def}, which is composed of the oracles 
    $U$ and $U'$ to produce $\ket{\bar{G}}$ and $\ket{\bar{G}'}$; this is followed by the post-selection 
    operation to obtain $\ket{G}\ket{G'}$ and the swap test. 
    The probability of obtaining $0$ as a result of the measurement on the first qubit is 
    $\Pr(0)=(1+|\braket{G|G'}|^2)/2$, which enables us to estimate the inner product 
    \eqref{eq:kernel_def}.} 
    \label{qc:swap-test}
\end{figurehere}

\subsection*{Alternative algorithm with switch test}

\begin{figurehere}
    \begin{align}
        \Qcircuit @C=1em @R=.7em {
            \lstick{\ket{0}}   & \qw     & \multigate{2}{\textrm{Oracle}} & \ctrlo{1}            & \ctrl{1}              & \qw    & \qw \\
            \lstick{\ket{0^n}} & {/} \qw & \ghost{\textrm{Oracle}}        & \gate{H^{\otimes n}} & \gate{H^{\otimes n'}} & \meter & \qw \\
            \lstick{\ket{0..}} & {/} \qw & \ghost{\textrm{Oracle}}        & \qw                  & \qw                   & \qw    & \qw
        }
    \end{align}
    \caption{Quantum circuit to prepare the state \eqref{eq:sp_g_without_index}. }
    \label{qc:init-index-reg-sp}
\end{figurehere} \noindent
\\

Here, we show that the use of the superposition 
\begin{equation}
\label{eq:sp_g_without_index}
    \ket{0}\ket{G}+\ket{1}\ket{G'} 
\end{equation}
can also be used to compute the inner product \eqref{eq:kernel_def} which eventually yields 
a different kernel than $K_{\rm BH}(G, G')$, yet using the same order of queries as the previous 
case. 
Similar to the previous case, the superposition \eqref{eq:sp_g_without_index} can be effectively 
generated by the post-selection operation on 
\begin{equation}
\label{eq:sp_g_with_index}
      \ket{0}\ket{\bar{G}}+\ket{1}\ket{\bar{G}'}, 
\end{equation}
using the circuit depicted in Fig.~\ref{qc:init-index-reg-sp}. 
We first apply the controlled \(H^{\otimes n}\) and \(H^{\otimes n'}\) on the 
index state of Eq.~\eqref{eq:sp_g_with_index} and then post-select the state when the measurement 
result on the index is \( 0^n \). 
The formal statement of the result in terms of the general discrete functions $f$ and $g$ is given 
as follows 
(the proof is given in the Supplementary Information):

\begin{theorem} \label{thm:rir-qf}
Let \(Y\) and \(Y'\) be the ranges of \(f\) and \(g\), respectively, and let \(a=|Y|\) and \(a'=|Y'|\). Thus, \(Y=\{y_k\}_{k=1}^a\) and \(Y'=\{y_k'\}_{k=1}^{a'}\) with \(y_k\) and \(y_k'\) elements of \(Y\) 
and \(Y'\), respectively. 
Also, let \(X_k=\{x \, |\, f(x)=y_k,x\in\{0,1\}^n\}\) and \(X_k'=\{x \, |\, g(x)=y_k',x\in\{0,1\}^{n'}\}\) for 
\(n\geq n'\). 
The quantum circuit depicted in Fig.~\ref{qc:init-index-reg-sp} uses the oracle to prepare the state 
\begin{align}
        \frac{1}{\sqrt{2}}&\left(\frac{\ket{0}\sum_{x\in\{0,1\}^{n}}\ket{x}\ket{f(x)}}{\sqrt{2^{n}}} \right. \\
        &+ \left. \frac{\ket{1}\ket{0^{n-n'}}\sum_{x\in\{0,1\}^{n'}}\ket{x}\ket{g(x)}}{\sqrt{2^{n'}}}\right). \label{eq:thm2-init-state}
\end{align}
Then, the final feature state of the circuit, which is post-selected when the measurement result is $0^n$ 
in the index state, is given by 
\begin{equation}
        \frac{1}{N_{fg}}\left(\frac{\ket{0}\sum_{x\in\{0,1\}^{n}}\ket{f(x)}}{2^{n}}
            +\frac{\ket{1}\sum_{x\in\{0,1\}^{n'}}\ket{g(x)}}{2^{n'}}\right), 
\label{eq:thm2-final-state}
\end{equation}
where
\begin{equation}
        N_{fg}=\sqrt{\left(\frac{\sum_{k=1}^a|X_{k}|^2}{2^{2n}}
                  +\frac{\sum_{k=1}^{a'}|X_{k}'|^2}{2^{2n'}}\right)}.
\end{equation}
The probability of obtaining the state \eqref{eq:thm2-final-state} is at least \((a^{-1}+a'^{-1})/2\). 
\end{theorem}
\mbox{}

As in the previous case, by imposing the encoding functions $E(G, x)$ and $E(G', x)$ to satisfy 
the condition \eqref{eq:v_range_constraint}, the quantum circuit depicted in Fig.~\ref{qc:init-index-reg-sp} 
transforms Eq.~\eqref{eq:sp_g_with_index} to Eq.~\eqref{eq:sp_g_without_index} with a probability of 
$\Omega(\sqrt{n}/n^c + \sqrt{n'}/n'\mbox{}^c) = \Omega(\sqrt{n}/n^c)$, where $n$ and $n'$ are 
assumed to be of the same order. 
The proof of this result is the same as that of Theorem \ref{thm:rir-v-range}.

\begin{figurehere}
    \begin{align}
        \Qcircuit @C=1em @R=.7em {
            \lstick{\ket{0}}   & \gate{H} & \ctrlo{1}        & \ctrl{1}          & \qw    & \gate{H} & \meter \\
            \lstick{\ket{0^n}} & {/} \qw  & \multigate{1}{U} & \multigate{1}{U'} & \meter & \qw      & \qw \\
            \lstick{\ket{0..}} & {/} \qw  & \ghost{U}        & \ghost{U'}        & \qw    & \qw      & \qw
        }
    \end{align}
    \caption{
    Quantum circuit to compute the inner product \eqref{eq:kernel_def}, which is composed of the oracles 
    $U$ and $U'$ to produce $\ket{0}\ket{\bar{G}}+\ket{1}\ket{\bar{G}'}$. This is followed by the 
    post-selection operation to obtain $\ket{0}\ket{G}+\ket{1}\ket{G'}$ and the switch test. 
    Additionally, \(n\geq n'\) is assumed. 
    The probability of obtaining $0$ as a result of a measurement on the first qubit is 
    $\Pr(0)=(1+\braket{G|G'})/2$, which enables us to estimate the inner product \eqref{eq:kernel_def}. }
    \label{qc:switch-test}
\end{figurehere}
\mbox{}

We have now obtained the state $\ket{0}\ket{G}+\ket{1}\ket{G'}$, which enables the application 
of the switch test \cite{SWITCH-test-1st,SWITCH-test-called} to compute the inner product 
\eqref{eq:kernel_def}. 
The circuit diagram, which contains the post-selection operation, is shown in Fig. \ref{qc:switch-test}. 
The total query complexity to compute the inner product \eqref{eq:kernel_def} is $O(n^c/\sqrt{n})$. 
The resulting inner product computed through the switch test, which we call the SH kernel, is given by 
\begin{align}
\label{eq:kernel_def_SH}
    K_{\rm SH}(G, G') = k_{\rm SH}(G, G') f_G^\top f_{G'},
\end{align}
where, again, $f_G=[|X_1|, \ldots, |X_a|]^\top$, and the coefficient $k_{\rm SH}(G, G')$ is given by 
\begin{equation}
\label{eq:switch_test_coeff}
    k_{\rm SH}(G, G')
       =\frac{2}{\frac{2^{n'}}{2^n}\sum_{k=1}^a |X_k|^2 + \frac{2^n}{2^{n'}}\sum_{k=1}^a |X_k'|^2}.
\end{equation}
In the Supplementary Information, we prove that this is a positive semidefinite kernel. 
Also we show there that $K_{\rm SH}(G, G') \leq K_{\rm BH}(G, G')$ holds, implying that the SH 
kernel may give a conservative classification performance than BH. 
Note that the generalized T-student kernel \cite{GeneralizedTStudentKernelUsage} has a similar form.


\subsection*{Improved algorithm with amplitude amplification}

\begin{figurehere}
    \begin{align}
        \Qcircuit @C=1em @R=.7em {
            \lstick{\ket{0^n}\ket{0..}} & {/} \qw & \gate{\textrm{Oracle}} & \qw \gategroup{1}{3}{1}{3}{.7em}{--} \\
                                        &         & \push{\text{AA}}
        } \\ \\
        \Qcircuit @C=1em @R=.7em {
            \lstick{\ket{0}}            & \qw     & \multigate{1}{\textrm{Oracle}} & \qw \\
            \lstick{\ket{0^n}\ket{0..}} & {/} \qw & \ghost{\textrm{Oracle}}        & \qw \gategroup{1}{3}{2}{3}{.7em}{--} \\
                                        &         & \push{\text{AA}}
        }
    \end{align}
    \caption{Quantum circuit with amplitude amplification to prepare the state 
    \eqref{eq:g_without_index} (upper) and \eqref{eq:sp_g_without_index} (lower). 
    The circuit that is iterated to realize the amplitude amplification is enclosed in the dashed line.}
    \label{qc:g-grover}
\end{figurehere}
\mbox{}
\\

Recall that we used post-selection on the state \eqref{eq:has-other-terms}: 
\begin{equation*} 
        \frac{1}{2^n}\ket{0^n}\left(\sum_{k=1}^a |X_k|\ket{y_k}\right)+\mathrm{other\ terms}, 
\end{equation*}
to probabilistically produce the first term 
\begin{align}
\label{target state AA}
    \frac{1}{2^n}\ket{0^n}\left(\sum_{k=1}^a |X_k|\ket{y_k}\right).
\end{align}
We can enhance the first term using the amplitude amplification (AA) operation 
\cite{Grover,AmplitudeAmplification}) to {\it deterministically} produce the same state 
\eqref{target state AA}. 
The clear advantage of AA is that it requires the square root of the number of operations 
to obtain this state. This is preferable over the previous repeat-until-success strategy. 
This means that the query complexity to transform Eq.~\eqref{eq:g_with_index} to 
Eq.~\eqref{eq:g_without_index} via AA is 
\begin{align}
    O\left(\sqrt{a/\sqrt{n}}\right)=O(\sqrt{a}/n^{1/4}). 
\end{align}
Similarly, transforming Eq. \eqref{eq:sp_g_with_index} to Eq. \eqref{eq:sp_g_without_index} 
via AA requires a query complexity of $O(\sqrt{a}/n^{1/4})$. 
These circuits are depicted in Fig. \ref{qc:g-grover}; note that the measurement on the index 
state is not necessary.

The circuit that includes the swap test and AA to compute the inner product 
$\braket{G|G'} = K_{\rm BH}(G, G')$ is depicted in Fig.~\ref{qc:swap-switch-test-grover} (upper). 
The circuit length is 
\begin{align}
    O(\sqrt{a}/n^{1/4})+O(\sqrt{a}/n^{1/4})=O(\sqrt{a}/n^{1/4}).
\end{align}
Also, the circuit containing the switch test and AA is depicted in Fig.~\ref{qc:swap-switch-test-grover} 
(lower); as in the case of swap test, the circuit length is $O(\sqrt{a}/n^{1/4})$. 
In particular, if the encoding functions $E(G, x)$ and $E(G', x)$ satisfy the condition 
\eqref{eq:v_range_constraint}, then we can specify $a=O(n^c)$.

\begin{figurehere}
    \begin{align}
        \Qcircuit @C=1em @R=.7em {
            \lstick{\ket{0}}               & \qw     & \qw             & \gate{H} & \ctrl{3} & \gate{H} & \meter \\
            \lstick{\ket{0^n}\ket{0..}}    & {/} \qw & \gate{U}        & \qw      & \qswap   & \qw      & \qw    \\
                                           &         & \push{\text{AA}} & & & & \\
            \lstick{\ket{0^{n'}}\ket{0..}} & {/} \qw & \gate{U'}       & \qw      & \qswap   & \qw      & \qw \gategroup{2}{3}{2}{3}{.7em}{--} \gategroup{4}{3}{4}{3}{.7em}{--} \\
                                           &         & \push{\text{AA}}
        } \\ \\
        \Qcircuit @C=1em @R=.7em {
            \lstick{\ket{0}}            & \gate{H} & \ctrlo{1} & \ctrl{1}  & \gate{H} & \meter \\
            \lstick{\ket{0^n}\ket{0..}} & {/} \qw  & \gate{U}  & \gate{U'} & \qw      & \qw \gategroup{1}{2}{2}{4}{.7em}{--}
        } \\
        \Qcircuit @C=1em @R=.7em {
            \push{\text{AA}} & & & & & & & &
        }
    \end{align}
    \caption{
    Quantum circuit to compute the inner product \eqref{eq:kernel_def}; it is composed of the oracles 
    enhanced via AA followed by the swap test (upper) and the switch test (lower).}
    \label{qc:swap-switch-test-grover}
\end{figurehere}

\subsection*{Time complexity for specific encoding functions}

Here, we discuss the time complexity that takes into account the number of elementary operations 
contained in the oracle. 
We particularly investigate the following two encoding functions: 
\begin{align}
    E_{ve}(G,x)&=[\#v,\#e], \label{eq:encode_ve} \\
    E_{ved}(G,x)&=[\#v,\#e,\#d1,\#d2,\#d3], 
\label{eq:encode_ved}
\end{align}
where \(\#v\) is the number of vertices of the subgraph specified by $x$, and \(\#e\) is the number 
of edges of $x$. 
Additionally, \(\#dD\) \((D\in\{1,2,3\})\) denotes the number of vertices that have a degree of \(D\). 
Then, from Lemmas \ref{thm:v}, \ref{thm:e}, and \ref{thm:dD} given in the Supplementary Information, 
the time complexities required to calculate Eq. \eqref{eq:encode_ve} and Eq. \eqref{eq:encode_ved} 
are 
\begin{align}
    O(n(\log n)^2)+O(|E|(\log |E|)^2) 
    =O((n+|E|)(\log n)^2),
\end{align}
and 
\begin{align}
    &O(n(\log n)^2)+O(|E|(\log |E|)^2) \\
    & \hspace{1em} +3\cdot O(n((\log n)^2+d(\log d)^2)) \\
    &=O((n+|E|)(\log n)^2+nd(\log d)^2), 
\end{align}
respectively. 
Recall that $d$ is the maximum degree of the graph $G$. 
The cardinality \(a=|Y|\) for the range \(Y\) can be evaluated as 
\begin{align}
    a=O(n)\cdot O(n^2)=O(n^3),
\end{align}
in the case of Eq. \eqref{eq:encode_ve}, and 
\begin{align}
    a=O(n)\cdot O(n^2)\cdot O(n)\cdot O(n)\cdot O(n)=O(n^6),
\end{align}
in the case of Eq. \eqref{eq:encode_ved}. 
Note that Eqs. \eqref{eq:encode_ve} and \eqref{eq:encode_ved} satisfy the condition 
\eqref{eq:v_range_constraint}. 
Hence, the time complexity of the quantum algorithm, assisted by AA, is evaluated as follows: 
in the case of Eq. \eqref{eq:encode_ve}, it is
\begin{equation}
    O((n+|E|)(\log n)^2)\cdot O\left(\sqrt{n^3}/n^{1/4}\right) 
      =O(n^{3.25}(\log n)^2),
\end{equation}
and in the case of Eq. \eqref{eq:encode_ved}, it is 
\begin{align}
    &O((n+|E|)(\log n)^2+nd(\log d)^2)\cdot O\left(\sqrt{n^6}/n^{1/4}\right) \\
    &\hspace{1em} =O(n^{4.75}(\log n)^2),
\end{align}
where \(|E|=n^2\) and \(d=n\) are used.
In contrast, the time complexity of typical existing graph kernels is \(O(n^3)\) for the random walk method \cite{FastVishwanathan} and \(O(n^4)\) for the shortest 
paths method \cite{Borgwardt}. 
Hence, our quantum computing approach for kernel computation has a time complexity 
comparable to that of the typical classical approach. 
However, note again that our kernel reflects features from {\it all} subgraphs, which are not 
covered by existing methods.

\subsection*{Toy example}

\begin{figurehere}
    \centering
    \includegraphics[width=6cm]{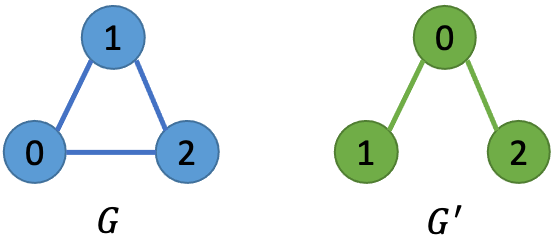}
    \caption{Structure of the toy graphs. }
    \label{fig:example-graphs}
\end{figurehere}
\mbox{}
\\

We consider two simple toy graphs, which are depicted in Fig. \ref{fig:example-graphs}, to 
demonstrate how to construct the corresponding quantum feature states $\ket{G}$ and $\ket{G'}$. 
The encoding function is chosen as $E_{ved}(G, x)$, given in Eq. \eqref{eq:encode_ved}. 
Note that both graphs have \(2^3=8\) induced subgraphs; thus, we need $n=3$ qubits 
to cover all subgraphs.

First, $\ket{\bar{G}}\in {\cal H}_{\rm index}\otimes {\cal H}_{\rm feature}$ is constructed as 
\begin{align}
    \ket{\bar{G}}&=\ket{000}\ket{0,0,0,0,0} + \ket{100}\ket{1,0,0,0,0} \\
      &+ \ket{010}\ket{1,0,0,0,0} + \ket{001}\ket{1,0,0,0,0} \\
      &+ \ket{110}\ket{2,1,2,0,0} + \ket{101}\ket{2,1,2,0,0} \\
      &+ \ket{011}\ket{2,1,2,0,0} + \ket{111}\ket{3,3,0,3,0}. 
\end{align}
For example, the term \(\ket{110}\ket{2,1,2,0,0}\) represents the state of the subgraph composed 
of the 0th and 1st vertices (thus, the index is 110); 
this subgraph has 2 vertices, 1 edge, 2 vertices with a degree of 1, 0 vertices with a degree of 2, 
and 0 vertices with a degree of 3 (thus, the feature is represented by $2,1,2,0,0$). 
Note again that the normalization constant is omitted. 
Therefore, the algorithm depicted in Fig.~\ref{qc:init-index-reg} enables us to remove the index state and 
arrive at the feature state $\ket{G}\in {\cal H}_{\rm feature}$: 
\begin{align}
    \ket{G}&=\ket{0,0,0,0,0}+3\ket{1,0,0,0,0} \\
       &+3\ket{2,1,2,0,0}+\ket{3,3,0,3,0}.
\end{align}
The state $\ket{\bar{G}'}$ can also be obtained in the same way as 
\begin{align}
    \ket{\bar{G}'}&=\ket{000}\ket{0,0,0,0,0} + \ket{100}\ket{1,0,0,0,0} \\
      &+ \ket{010}\ket{1,0,0,0,0} + \ket{001}\ket{1,0,0,0,0} + \\
      &+ \ket{110}\ket{2,1,2,0,0} + \ket{101}\ket{2,1,2,0,0} + \\
      &+ \ket{011}\ket{2,0,0,0,0} + \ket{111}\ket{3,2,2,1,0}, 
\end{align}
which leads to 
\begin{align}
    \ket{G'}&=\ket{0,0,0,0,0}+3\ket{1,0,0,0,0} \\
    &+2\ket{2,1,2,0,0}+\ket{2,0,0,0,0}+\ket{3,2,2,1,0}.
\end{align}
Hence, by considering the normalization factor, the inner product (i.e., the similarity 
of the graphs) is calculated as follows: 
in the case of the swap test, it is 
\begin{align}
    & \hspace{-1em} \braket{G|G'} = K_{\rm BH}(G, G')\\
    &=\frac{1\cdot 1 + 3\cdot 3 + 3\cdot 2 + 0\cdot 1 + 1\cdot 0 + 0\cdot 1}{\sqrt{1^2+3^2+3^2+1^2}\sqrt{1^2+3^2+2^2+1^2+1^2}} \\
    &\sim 0.8944,
\end{align}
and in the case of the switch test, it is 
\begin{align}
    &\hspace{-1em} \braket{G|G'} = K_{\rm SH}(G, G')\\
    &=\frac{2(1\cdot 1 + 3\cdot 3 + 3\cdot 2 + 0\cdot 1 + 1\cdot 0 + 0\cdot 1)}{\frac{2^3}{2^3}(1^2+3^2+3^2+1^2)+\frac{2^3}{2^3}(1^2+3^2+2^2+1^2+1^2)} \\
    &\sim 0.8889.
\end{align}
Note that they are certainly bigger than the value $\braket{\bar{G}|\bar{G}'}=0.75$ computed 
via the naive method.

\subsection*{Advantage over classical sampling method}

In the classical case, an exponentially large resource is necessary to compute a feature vector, 
corresponding to Eq.~\eqref{eq:g_without_index}, which considers all subgraphs; however, we can 
utilize an efficient classical sampling method that was used in \cite{N-Shervashidze} to approximate 
this feature vector. 
More specifically, we first sample $S$ subgraphs from the entire graph $G$ and then apply the 
encoding function $E(G,x)$, which satisfies Eq. \eqref{eq:v_range_constraint}, on those subgraphs 
to approximate the feature vector. 
The $k$th component of this vector is given by 
\begin{align}
        \hat{P}_{\mathbf{x}^S}(y_k)=\frac{1}{S}\sum_{i=1}^S 1( E(G, x_i)=y_k),
\end{align}
where $\mathbf{x}^S=\{ x_1,...,x_S\}$ are the set of indices that identify the sampled subgraphs. 
Also, $1(A)$ represents the indicator function, which takes 1 when the condition $A$ is satisfied 
and zero otherwise. 
Now, the $k$th component of the true feature vector is represented as
\begin{align} 
\label{eq:true-prob-dist}
        P(y_k)=\frac{|X_k|}{2^n}.
\end{align}
We then have the following theorem (the proof is given in the Supplementary Information):

\begin{theorem} 
\label{thm:classical-sample-size}
For a given $\epsilon>0$ and $\delta>0$, 
\begin{align}
    S=O\left(\frac{a-\log{\delta}}{\epsilon^2\log{n}}\right)
\end{align}
samples suffice to ensure that 
\begin{align} 
\label{eq:cl-l1-bound}
        \Pr\left(\left\|P-\hat{P}_{\mathbf{x}^S}\right\|_1\geq\epsilon\right)\leq\delta, 
\end{align}
where $\| \cdot \|_1$ denotes the $L_1$ norm. 
In particular, for constant $\epsilon$ and $\delta$, we have that $S=O(a/\log{n})$. 
\end{theorem}

Therefore, the sample complexity of this classical method is $O(a/\log{n})$, whereas the query 
complexity of the proposed quantum algorithm is $O(a/\sqrt{n})$. 
Hence, up to the difference of the sense of complexities, the proposed method has a clear 
computational  advantage. 
Note that the inner product $\braket{G|G'}$ that is computed using the above classical sampling 
method is given by 
\begin{align}
    \braket{G|G'} = \frac{1}{(\sum_{k=1}^a |X_k|)(\sum_{k=1}^a |X_k'|)} f_G^\top f_{G'}, 
\end{align}
which differs from that computed in the quantum case \eqref{eq:kernel_def_nff} or 
\eqref{eq:kernel_def_SH}.

\subsection*{Numerical experiment}

Here we study the performance of classifiers constructed based on the proposed quantum kernel, 
with comparison to some typical classical classifiers. 
The quantum kernel was calculated, not using the quantum algorithm but via the direct calculation 
of the inner product \eqref{eq:kernel_def}, in an ideal noise-free environment on a GPU; 
for the details, see the Code Availability section in Supplementary Information. 
We calculate both the BH kernel \eqref{eq:kernel_def_nff} and the SH kernel \eqref{eq:kernel_def_SH}, 
for the two different encoding functions $E_{ve}$ given in Eq.~\eqref{eq:encode_ve} and $E_{ved}$ 
given in Eq.~\eqref{eq:encode_ved}. 
We compare the quantum graph kernels to the following classical graph kernels: 
the random walk kernel (RW) \cite{FastVishwanathan}, the graphlet sampling kernel (GS) 
\cite{N-Shervashidze}, and the shortest path kernel (SP) \cite{Borgwardt}. 
These three classical kernels are simulated using Python's \textit{GraKeL} library \cite{grakel}.


We use the following benchmark datasets obtained from the repository of the Technical University 
of Dortmund \cite{datasets}. 
For the case of binary classification problems, we used: 
AIDS (chemical compounds with or without evidence of anti-HIV activity \cite{AIDS}); 
BZR\_MD (dataset BZR of active or inactive benzodiazepine receptors \cite{ER}; 
converted to complete graphs \cite{ER_MD}); 
ER\_MD (dataset ER of active or inactive estrogen receptors \cite{ER}; 
converted to complete graphs \cite{ER_MD}); 
IMDB-BINARY (the movie genre is action or romance based on its co-starring 
relationship \cite{IMDB}); 
MUTAG (chemical compounds with or without mutagenicity \cite{MUTAG}); 
and PTC\_FM (chemical compounds in the PTC dataset \cite{PTC} that are carcinogenic 
or non-carcinogenic to female mice \cite{ER_MD}). 
As for the multi-class classification problems, we used: 
15-classes Fingerprint (fingerprint images converted to graphs and divided by type \cite{Fingerprint}) 
and 
3-classes IMDB-MULTI (the movie genre is comedy, romance, or sci-fi based on its co-starring 
relationship \cite{IMDB}). 
Due to the limitation of the GPU memory, we took graphs with less than or equal to 28 vertices. 
As a result, (the number of graphs)/(the total number of graphs) are 
1774/2000 for AIDS, 
296/306 for BZR\_MD, 
398/446 for ER\_MD, 
860/1000 for IMDB-BINARY, 
188/188 for MUTAG, 
331/349 for PTC\_FM, 
2148/2148 (excluding graphs with \(\#edges\) less than 1) for Fingerprint, and 
1406/1500 for IMDB-MULTI.
The necessary number of qubits is $28+\log{28}+\log{(28\times 27/2)}\sim 41$ for the case 
$E_{ve}$ and $28+\log{28}+\log{(28\times 27/2)}+\log{28}+\log{28}+\log{28}\sim 56$ for the 
case $E_{ved}$.

We apply the $C$-support vector machine (SVM), implemented via \textit{Scikit-learn} 
\cite{scikit-learn}, to classify the dataset. 
To evaluate the classification performance, we calculate the mean test accuracy, by running 
10 repeats of a double 10-fold cross-validation. 
In addition, we calculate {\it F-measure}, which is used when the number of data in different 
classes are unbalanced; in fact, the numbers of data of two classes are 63 and 125 for MUTAG 
and 400 and 1600 for AIDS. 
The SVM parameter $C$ is taken from the discrete set $\{10^{-4}, 10^{-3}, \ldots, 10^3\}$, and 
the best model with respect to $C$ is used to compute the classification performance. 
The result are summarized in Table~\ref{tb:graphkernels}. 

\end{multicols}

\begin{table}[H]
    \centering
    \caption{
    Mean test accuracy (upper) and F-measure (lower) of the $C$-SVM constructed with each kernel; 
    the errors are the standard deviation between 10 repetitions of the double 10-fold cross-validation. 
    QK is the proposed quantum kernel. 
    RW, GS, and SP are the classical graph kernels. 
    Macro-F1 is used in the multiclass graph datasets Fingerprint and IMDB-MULTI. 
    The bold indicates the best performing value in the dataset.}
    \label{tb:graphkernels}
    \scalebox{0.8}[0.9]{
    \begin{tabular}{l|ccccccc} \hline
        Dataset & QK (BH\([ve]\)) & QK (BH\([ved]\)) & QK (SH\([ve]\)) & QK (SH\([ved]\)) & RW & GS & SP \\ \hline
        AIDS & \(\mathbf{99.79}\pm 0.06\) & \(99.68\pm 0.05\) & \(\mathbf{99.79}\pm 0.06\) & \(99.71\pm 0.05\) & \(99.66\pm 0.00\) & \(98.74\pm 0.16\) & \(99.67\pm 0.02\) \\
        BZR\_MD & \(64.23\pm 0.71\) & \(\mathbf{64.29}\pm 0.44\) & \(63.83\pm 1.01\) & \(63.83\pm 1.01\) & \(62.14\pm 0.76\) & \(53.90\pm 3.07\) & \(63.42\pm 1.02\) \\
        ER\_MD & \(66.05\pm 1.36\) & \(66.00\pm 1.35\) & \(\mathbf{66.28}\pm 1.05\) & \(\mathbf{66.28}\pm 1.05\) & \(60.73\pm 0.47\) & \(55.40\pm 1.40\) & \(65.58\pm 0.83\) \\
        IMDB-BINARY & \(\mathbf{70.16}\pm 0.90\) & \(69.85\pm 1.15\) & \(69.81\pm 0.60\) & \(69.85\pm 1.03\) & \(53.64\pm 0.72\) & \(42.05\pm 0.56\) & \(57.01\pm 1.09\) \\
        MUTAG & \(85.88\pm 0.59\) & \(87.01\pm 1.20\) & \(85.56\pm 0.73\) & \(86.79\pm 0.96\) & \(\mathbf{88.11}\pm 0.70\) & \(69.94\pm 1.30\) & \(86.73\pm 1.30\) \\
        PTC\_FM & \(\mathbf{60.82}\pm 1.30\) & \(60.12\pm 1.20\) & \(60.55\pm 1.24\) & \(60.15\pm 1.15\) & \(57.86\pm 0.99\) & \(52.29\pm 2.49\) & \(57.91\pm 0.74\) \\
        Fingerprint & \(46.85\pm 0.36\) & \(\mathbf{47.09}\pm 0.27\) & \(46.94\pm 0.33\) & \(47.01\pm 0.24\) & \(47.03\pm 0.29\) & \(42.82\pm 0.56\) & \(46.99\pm 0.29\) \\
        IMDB-MULTI & \(46.44\pm 0.23\) & \(47.13\pm 0.53\) & \(46.66\pm 0.48\) & \(\mathbf{47.57}\pm 0.48\) & \(35.49\pm 0.20\) & \(18.95\pm 0.34\) & \(42.53\pm 0.98\) \\
        \hline \hline
        AIDS & \(\mathbf{99.88}\pm 0.03\) & \(99.82\pm 0.03\) & \(\mathbf{99.88}\pm 0.03\) & \(99.84\pm 0.03\) & \(99.81\pm 0.00\) & \(99.30\pm 0.09\) & \(99.82\pm 0.01\) \\
        BZR\_MD & \(70.44\pm 0.55\) & \(\mathbf{70.46}\pm 0.42\) & \(69.87\pm 1.11\) & \(69.87\pm 1.11\) & \(62.28\pm 1.08\) & \(54.06\pm 3.27\) & \(66.70\pm 1.48\) \\
        ER\_MD & \(62.75\pm 1.75\) & \(62.73\pm 1.73\) & \(\mathbf{62.89}\pm 1.58\) & \(\mathbf{62.89}\pm 1.58\) & \(1.18\pm 2.52\) & \(41.53\pm 1.98\) & \(53.31\pm 1.95\) \\
        IMDB-BINARY & \(\mathbf{69.27}\pm 1.40\) & \(69.06\pm 1.60\) & \(68.47\pm 0.91\) & \(67.91\pm 1.23\) & \(25.06\pm 1.76\) & \(41.79\pm 0.77\) & \(66.87\pm 1.13\) \\
        MUTAG & \(89.15\pm 0.39\) & \(90.11\pm 0.86\) & \(88.96\pm 0.49\) & \(89.95\pm 0.64\) & \(\mathbf{90.74}\pm 0.54\) & \(76.85\pm 1.24\) & \(89.68\pm 0.95\) \\
        PTC\_FM & \(34.75\pm 1.81\) & \(36.40\pm 1.63\) & \(34.28\pm 2.03\) & \(35.16\pm 1.96\) & \(2.88\pm 4.42\) & \(\mathbf{40.99}\pm 3.08\) & \(2.72\pm 2.44\) \\
        Fingerprint & \(17.73\pm 0.17\) & \(17.81\pm 0.32\) & \(17.75\pm 0.13\) & \(\mathbf{17.84}\pm 0.08\) & \(17.32\pm 0.33\) & \(17.21\pm 0.22\) & \(16.74\pm 0.15\) \\
        IMDB-MULTI & \(44.67\pm 0.35\) & \(45.40\pm 0.61\) & \(44.91\pm 0.60\) & \(\mathbf{45.97}\pm 0.55\) & \(21.04\pm 0.37\) & \(17.17\pm 0.31\) & \(38.71\pm 1.31\) \\
        \hline
    \end{tabular}
    }
\end{table}

\begin{multicols}{2}

The table shows that, in many cases, the proposed quantum kernel achieves better classification 
accuracy than that obtained via the classical kernels. 
Note that the AIDS dataset is sparse and ER\_MD dataset is dense; the quantum kernels show the 
better performance in both cases, implying that they are not significantly affected by the density 
of the graph dataset. 
In many cases, the two kernels BH and SH show a similar performance, but there are some 
visible differences depending on the dataset; 
this may be due to the property $K_{\rm SH}(G, G') \leq K_{\rm BH}(G, G')$, which is proven in 
Lemma \ref{thm:cosinesimilarity-geq-ourkernel} in Supplementary Information. 
Also, although $E_{ved}$ has more features than $E_{ve}$, from the table we do not observe a clear 
superiority of the former over the latter in the classification accuracy; we will discuss further on this 
point in the next section.

Lastly let us check the probability for successfully removing the index register; 
recall from Theorem 2 that the lower bound is $\Omega(\sqrt{n}/a)$. 
Figure~\ref{fig:prob-change} shows the success probabilities when Eq. \eqref{eq:encode_ved} is 
used as the encoding function, in which case the lower bound is \(\Omega(\sqrt{n}/a) = \Omega(n^{-5.5})\). 
As shown in the figure, the actual success probability is much higher than the lower bound, indicating 
that the quantum algorithm for removing the index state will be more efficient than expected by 
the theory.

\begin{figurehere}
    \centering
    \includegraphics[width=\linewidth]{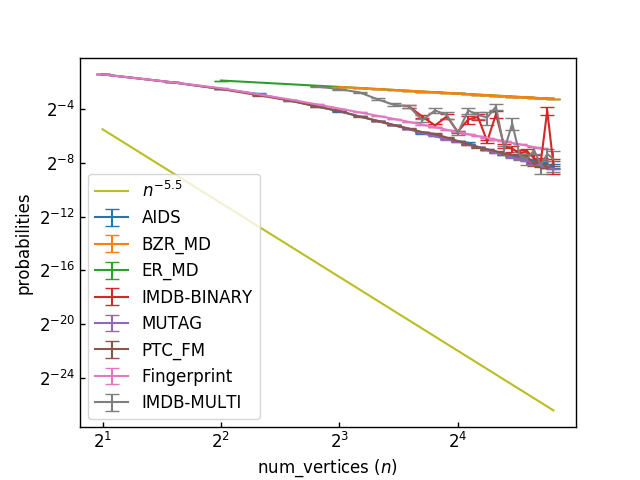}
    \caption{
    Success probability for removing the index state. 
    The horizontal axis represents the number of vertices $n$, and the vertical axis represents 
    the success probability. The error bar represents the standard error.}
     \label{fig:prob-change}
\end{figurehere}


\section*{Discussion}

In this paper, as a main result, we provided the condition and the protocol for removing 
the index state with polynomial query complexity. 
The encoding function $E(G,x)$ that extracts features from a subgraph $x$, given by 
Eq.~\eqref{eq:encode_ve} or \eqref{eq:encode_ved}, satisfies this condition, which allows 
us to construct the graph kernel that correctly reflects features of all subgraphs. 
We gave a proof-of-principle numerical demonstration to solve the problem of classifying 
various type of graph set containing graphs at most 28 vertices, via the quantum simulator 
composed of 41 or 56 qubits. 
The proposed algorithm that efficiently removes the index states will be useful in various 
other problems such as the task of counting the same words for text classification problems 
using the Bhattacharyya kernel \cite{BhattacharyyaKernelText}.

We here give a remark on the choice of $E(G,x)$. 
One would consider that $E(G,x)$ with more features including e.g. a cycle structure 
\cite{CyclesReview}, which thus has a bigger range of function, may lead to better classification 
performance, although it needs more query complexity for removing the index state and thereby 
constructing the kernel. 
(In particular, as is well known, if $E(G,x)$ and $x$ are one-to-one correspondence, we need 
an exponential order of query complexity to do this task.) 
However, the important fact revealed by the numerical demonstration is that such a bigger-range 
encoding may be not required; actually, we found that $E_{ve}(G,x)$ and $E_{ved}(G,x)$ lead to 
almost the same classification performance. 
This is presumably because the proposed kernel covers all subgraphs. 
Hence, a relatively simple $E(G,x)$ might be sufficient, which is the advantage of our quantum 
algorithm. 
In other words, existing kernels that do not cover all subgraphs may need to contain more features.

A disadvantage of the kernel method is that it needs heavy computational cost for calculating 
the inner products $\braket{G_i|G_j}$ for all pairs of $(G_i, G_j)$ contained in the training dataset, 
in order to construct the Gram matrix. 
The generalization of the switch test protocol given in Theorem~\ref{thm:rir-qf} may give a solution 
to this issue. 
That is, we could have an algorithm that generates a superposition 
$\ket{1}\ket{G_1}+\ket{2}\ket{G_2}+\ket{3}\ket{G_3}+\cdots$ 
and thereby efficiently construct the Gram matrix by some means. 
This direction is worth investigating, as it is the scheme demanded in the field of kernel-based 
quantum machine learning.

The algorithms posed in this paper are all difficult to implement on a near-term quantum device. 
A key approach may be to develop a valid relaxation method of the condition, because very precise 
computation of the kernel value may be not necessary.


\section*{Supplementary Information}

This supplementary information contains proofs of the theorems in the main text, properties of the 
SH kernel, some lemmas related for constructing $E(G,x)$, and an additional information for our 
numerical simulations.

\subsection*{Proof of theorems}

\begin{proof}[Proof of Theorem \ref{thm:rir-v-range}]
    Suppose that the set \(X_{v,h}\) satisfies
    \begin{align}
        &\bigcup_{h=1}^{|Y_v|}X_{v,h}=X_v, \\
        &\bigcup_{h=1}^{|Y_v|}\bigcup_{\substack{h'\neq h, \\ h'\in[1,|Y_v|]}}(X_{v,h}\cap X_{v,h'})=\emptyset.
    \end{align}
    We rewrite Eq. \eqref{eq:prob-to-init-ir} as
    \begin{align}
        \Pr(0^n)=\frac{1}{2^{2n}}\sum_{v=0}^n \sum_{h=1}^{|Y_v|} |X_{v,h}|^2.
    \end{align}
    To calculate the lower bound of
    \begin{equation}
        \sum_{h=1}^{|Y_v|} |X_{v,h}|^2
    \end{equation}
    subjected to the equality constraint
    \begin{align}
        \sum_{h=1}^{|Y_v|} |X_{v,h}|=|X_v|=\binom{n}{v},
    \end{align}
    we define the cost function 
    \begin{align}
        &L(|X_{v,1}|,\dots,|X_{v,|Y_v|}|,\lambda) \\
        &=-\sum_{h=1}^{|Y_v|} |X_{v,h}|^2-\lambda\left(\sum_{h=1}^{|Y_v|}|X_{v,h}|-\binom{n}{v}\right),
    \end{align}
    where \(\lambda\) denotes the Lagrange multiplier. Clearly,
    \begin{equation}
        |X_{v,h}|=\frac{1}{|Y_v|}\binom{n}{v}
    \end{equation}
    maximizes \(L(|X_{v,1}|,\dots,|X_{v,|Y_v|}|,\lambda)\). 
    Thus, the probability that we obtain \(0^{n}\) when measuring the index state is evaluated as 
    \begin{align}
        \Pr(0^n)
        &\geq \frac{1}{2^{2n}}\sum_{v=0}^n \sum_{h=1}^{|Y_v|}\left(\frac{1}{|Y_v|}\binom{n}{v}\right)^2 
        =\frac{1}{2^{2n}}\sum_{v=0}^n \frac{1}{|Y_v|} \binom{n}{v}^2 \\
        &\geq\frac{1}{2^{2n}}\sum_{v=0}^n \frac{1}{\beta n^c} \binom{n}{v}^2 
        =\frac{1}{\beta n^c 2^{2n}}\binom{2n}{n} \\
        &\sim \frac{1}{\beta n^c 2^{2n}}\frac{2^{2n}}{\sqrt{\pi n}} 
        =\Omega\left(\frac{\sqrt{n}}{n^c}\right), 
    \end{align}
    where we used Eq.~\eqref{eq:order of |Y|} to have $|Y_v|={\rm Pol}(v^{c-1})\leq \beta n^c$ with 
    some constant $\beta$. 
    Note that $\Omega(\cdot)$ is defined as $p(n)=\Omega(q(n))$ through two probability distributions 
    $p$ and $q$ satisfying 
    \begin{align}
        \exists n_0, \exists M>0\ \mathrm{s.t.}\ n\geq n_0 \Rightarrow p(n)\geq Mq(n),
    \end{align}
    for real valued functions $p(n)$ and $q(n)$. 
\end{proof}

\begin{proof}[Proof of Theorem \ref{thm:rir-qf}]
    According to Fig.~\ref{qc:init-index-reg-sp}, we perform the controlled-\(H^{\otimes n}\) followed 
    by the controlled-\(H^{\otimes n'}\) on the following state
    \begin{align}
        &\frac{1}{\sqrt{2}}\left(\frac{\ket{0}\sum_{x\in\{0,1\}^{n}}\ket{x}\ket{f(x)}}{\sqrt{2^{n}}} \right. \\
        \quad &+ \left. \frac{\ket{1}\ket{0^{n-n'}}\sum_{x\in\{0,1\}^{n'}}\ket{x}\ket{g(x)}}{\sqrt{2^{n'}}}\right) \\
        &=\frac{1}{\sqrt{2}}\left(\frac{\ket{0}\sum_{k=1}^{a}\sum_{x\in X_{k}}\ket{x}\ket{y_{k}}}{\sqrt{2^{n}}} \right. \\
        \quad &+ \left. \frac{\ket{1}\ket{0^{n-n'}}\sum_{k=1}^{a'}\sum_{x\in{X_{k}'}}\ket{x}\ket{y_{k}'}}{\sqrt{2^{n'}}}\right).
    \end{align}
    Then, we have
    \begin{align}
        &\frac{1}{\sqrt{2}}\left\{\ket{0}\left(\frac{1}{2^{n}}\ket{0^{n}}\sum_{k=1}^{a}|X_{k}|\ket{y_{k}}+\mathrm{{other\ terms}}\right)\right. \\
        &+ \left. \ket{1}\left(\frac{1}{2^{n'}}\ket{0^{n}}\sum_{k=1}^{a'}|X_{k}'|\ket{y_{k}'}+\mathrm{{other\ terms}'}\right)\right\} \label{eq:sp-has-other-terms}.
    \end{align}
    The probability that we obtain \(0^{n}\) via the index state measurement is given by 
    \begin{equation} \label{eq:sp-prob-lower-bound}
        \Pr(0^{n})=\frac{1}{2}\left(\frac{\sum_{k=1}^{a}|X_{k}|^2}{2^{2n}}+\frac{\sum_{k=1}^{a'}|X_{k}'|^2}{2^{2n'}}\right).
    \end{equation}
    Owing to Eq. \eqref{eq:prob-lower-bound}, this probability is lower bounded by 
    \begin{equation}
        \Pr(0^{n})\geq \frac{1}{2}\left(a^{-1}+a'^{-1}\right).
    \end{equation}
    Therefore, we can remove the index state with probability at least \((a^{-1}+a'^{-1})/2\).
\end{proof}

\begin{proof}[Proof of Theorem \ref{thm:classical-sample-size}]
    We use the result given in \cite{Weissman}; 
    for the empirical distribution of a sequence of independent identically distributed random variables, 
    $\hat{P}_{\mathbf{x}^S}$, and the true distribution $P$, the following inequality holds: 
    \begin{align} \label{eq:l1-deviation}
        \Pr\left(\left\|P-\hat{P}_{\mathbf{x}^S}\right\|_1\geq\epsilon\right)
          &\leq(2^a-2)e^{-S\varphi(\pi_P)\epsilon^2/4},
    \end{align}
    where \(\varphi(p)\) and \(\pi_P\) are given by 
    \begin{equation}
        \varphi(p)=\frac{1}{1-2p}\ln{\frac{1-p}{p}}
    \end{equation}
    and
    \begin{equation}
        \pi_P=\max_{k} \min \{P(y_k), 1-P(y_k)\} = \max_{k} P(y_k).
    \end{equation}
    Here we assumed $P(y_k) < 1/2$, which in fact holds in our case. 
    Note that $\delta$ in Eq.~\eqref{eq:cl-l1-bound} is defined by the rightmost side of 
    Eq.~\eqref{eq:l1-deviation}. 
    Now, because the true probability distribution $P$ is given by Eq.~\eqref{eq:true-prob-dist}, we obtain
    \begin{equation}
        \pi_P=\max_{k}\frac{|X_k|}{2^n}.
    \end{equation}
    Then, from Eq.~\eqref{eq:v_range_constraint}, we have the following inequality: 
    \begin{align}
        \frac{1}{(n/2)^{c-1}}\binom{n}{n/2}\frac{1}{2^n} \leq \pi_P \leq \binom{n}{n/2}\frac{1}{2^n}, 
    \end{align}
    which implies 
    \begin{align}
        \frac{2^{c-1/2}}{\sqrt{\pi}n^{c-1/2}} \lesssim \pi_P \lesssim \frac{\sqrt{2}}{\sqrt{\pi n}}.
    \end{align}
    Thus, $\varphi(\pi_P)$ is of the order of $\Omega(\ln{n})$, and then $\delta$ can be evaluated 
    as $\delta \sim 2^a {\rm exp}(-S\epsilon^2 \Omega(\ln{n}))$, from Eq.~\eqref{eq:l1-deviation}. 
    As a result, we find 
    \begin{align}
        S=\frac{a\ln{2}-\ln{\delta}}{\epsilon^2 \Omega(\ln{n})} 
        =O\left(\frac{a-\log{\delta}}{\epsilon^2\log{n}}\right) 
        =O\left(\frac{a}{\log{n}}\right).
    \end{align}
    Therefore we arrive at \(S=O(a/\log{n})\).
\end{proof}

\subsection*{Properties of the SH kernel}

\subsubsection*{Positive semidefiniteness}

First, we prove that the SH kernel  is positive semidefinite, which is necessary to construct a 
valid classifier based on the SH kernel.

\begin{lemma}\label{lem:semidef}
    Let $\{ G_x, ~ x=1, \ldots, N\}$ be a set of graphs. 
    Then, the SH kernel \eqref{eq:kernel_def_SH}: 
    \begin{align}
        K_{\rm SH}(G, G')=k_{\rm SH}(G, G') f_{G}^\top f_{G'}, 
    \end{align}
    where $f_{G}=[|X_1|, \ldots, |X_a| ]^\top$ and
    \begin{align}
       k_{\rm SH}(G, G')
        =\frac{2}{\frac{2^{n'}}{2^n}\sum_{k=1}^a |X_k|^2 + \frac{2^n}{2^{n'}}\sum_{k=1}^a |X_k'|^2}, 
    \end{align}
    is a positive semidefinite kernel; 
    that is, the matrix $(K_{\rm SH}(G_x, G_y))_{x,y=1,\ldots,N}$ is a positive semidefinite matrix.  
\end{lemma}

\begin{proof}
    We use the following general fact; if $\kappa_1$ and $\kappa_2$ are positive semidefinite kernels, 
    then the product $\kappa(G,G')=\kappa_1(G,G')\kappa_2(G,G')$ is also a positive semidefinite kernel. 
    Now, because \(f_{G}^\top f_{G'}\) is positive semidefinite, we prove 
    that $k_{\rm SH}(G, G')$ is positive semidefinite. 
    For this purpose, we define
    \begin{align}
        p_x = \sum_{k=1}^a |X_{kx}|^2, ~~
        q_{x,y} = \frac{2^{n_y}}{2^{n_x}}.
    \end{align}
    Then,
    \begin{align}
        k_{\rm SH}(G_x, G_y)=\frac{2}{q_{x,y}p_x+q_{y,x}p_y}.
    \end{align}
    Also we define
    \begin{align} \label{eq:cholesky-L}
        L_{i,j}=\begin{cases}
            \frac{\prod_{k=1}^{j-1}\left(q_{k,j}p_k-q_{j,k}p_j\right)}
                   {\sqrt{2p_j}\prod_{k=1}^{j-1}\left(q_{k,j}p_k+q_{j,k}p_j\right)} & (i=j) \\
            \frac{\sqrt{2p_j}\prod_{k=1}^{j-1}\left(q_{k,i}p_k-q_{i,k}p_i\right)}
                   {\prod_{k=1}^j\left(q_{k,i}p_k+q_{i,k}p_i\right)} & (i>j) \\
            0 & (i<j).
        \end{cases}
    \end{align}    
    Below we will show that the matrix $(k_{\rm SH}(G_x, G_y))_{x,y=1,\ldots,N}$ is represented as 
    \begin{align} \label{eq:cholesky}
        k_{\rm SH} = 2 LL^\top,
    \end{align}
    meaning that $k_{\rm SH}(G, G')$ is positive semidefinite. 
    The proof is divided into the three cases (i), (ii), and (iii). 
    \\
    (i) The case \(x=y\). 
    \begin{align}
        \left(LL^\top\right)_{x,x}&=\sum_{l=1}^x L_{x,l}^2 \\
        &=\sum_{l=1}^{x-1}\frac{2p_l\prod_{k=1}^{l-1}\left(q_{k,x}p_k-\frac{1}{q_{k,x}}p_x\right)^2}{\prod_{k=1}^l\left(q_{k,x}p_k+\frac{1}{q_{k,x}}p_x\right)^2} \\
        &+\frac{\prod_{k=1}^{x-1}\left(q_{k,x}p_k-\frac{1}{q_{k,x}}p_x\right)^2}{2p_x\prod_{k=1}^{x-1}\left(q_{k,x}p_k+\frac{1}{q_{k,x}}p_x\right)^2}.
    \end{align}
    Here, we define \(\alpha(t)\), \(\beta(t)\) (\(t\in[0,x-1]\)) as
    \begin{align}
        \alpha(t)&=\sum_{l=1}^{t}\frac{2p_l\prod_{k=1}^{l-1}\left(q_{k,x}p_k-\frac{1}{q_{k,x}}p_x\right)^2}{\prod_{k=1}^l\left(q_{k,x}p_k+\frac{1}{q_{k,x}}p_x\right)^2}, \\
        \beta(t)&=\frac{\prod_{k=1}^{t}\left(q_{k,x}p_k-\frac{1}{q_{k,x}}p_x\right)^2}{\prod_{k=1}^{t}\left(q_{k,x}p_k+\frac{1}{q_{k,x}}p_x\right)^2}.
    \end{align}
    Then we have
    \begin{align}
        \left(LL^\top\right)_{x,x} = \alpha(x-1)+\frac{\beta(x-1)}{2p_x}.
    \end{align}
    When \(t\geq 1\),
    \begin{align}
        &\alpha(t)+\frac{\beta(t)}{2p_x} \\
        &=\alpha(t-1)
              +\frac{2p_t\prod_{k=1}^{t-1}\left(q_{k,x}p_k-\frac{1}{q_{k,x}}p_x\right)^2}
                       {\prod_{k=1}^t\left(q_{k,x}p_k+\frac{1}{q_{k,x}}p_x\right)^2} +\frac{\beta(t)}{2p_x} \\
        &=\alpha(t-1)+\left(4p_t p_x + \left(q_{t,x}p_t-\frac{1}{q_{t,x}}p_x\right)^2\right) \\
        &\hspace{2em} 
            \times\frac{\prod_{k=1}^{t-1}\left(q_{k,x}p_k-\frac{1}{q_{k,x}}p_x\right)^2}
                             {2p_x\prod_{k=1}^t\left(q_{k,x}p_k+\frac{1}{q_{k,x}}p_x\right)^2} \\
        &=\alpha(t-1)+\frac{\prod_{k=1}^{t-1}\left(q_{k,x}p_k-\frac{1}{q_{k,x}}p_x\right)^2}{2p_x\prod_{k=1}^{t-1}\left(q_{k,x}p_k+\frac{1}{q_{k,x}}p_x\right)^2} \\
        &=\alpha(t-1)+\frac{\beta(t-1)}{2p_x}.
    \end{align}
    Thus, we obtain the following equation
    \begin{align}
        \left(LL^\top\right)_{x,x}=\alpha(0)+\frac{\beta(0)}{2p_x}=\frac{1}{2p_x}
               =k_{\rm SH}(G_x, G_x)/2.
    \end{align}
    (ii) The case \(x<y\). 
    \begin{align}
        \left(LL^\top\right)_{x,y}&=\sum_{l=1}^x L_{x,l}L_{y,l} \\
        &=\sum_{l=1}^{x-1}\frac{\sqrt{2p_l}\prod_{k=1}^{l-1}\left(q_{k,x}p_k-\frac{1}{q_{k,x}}p_x\right)}{\prod_{k=1}^l\left(q_{k,x}p_k+\frac{1}{q_{k,x}}p_x\right)} \\
        &\hspace{2em}\times\frac{\sqrt{2p_l}\prod_{k=1}^{l-1}\left(q_{k,y}p_k-\frac{1}{q_{k,y}}p_y\right)}{\prod_{k=1}^l\left(q_{k,y}p_k+\frac{1}{q_{k,y}}p_y\right)} \\
        &+\frac{\prod_{k=1}^{x-1}\left(q_{k,x}p_k-\frac{1}{q_{k,x}}p_x\right)}{\sqrt{2p_x}\prod_{k=1}^{x-1}\left(q_{k,x}p_k+\frac{1}{q_{k,x}}p_x\right)} \\
        &\hspace{2em}\times\frac{\sqrt{2p_x}\prod_{k=1}^{x-1}\left(q_{k,y}p_k-\frac{1}{q_{k,y}}p_y\right)}{\prod_{k=1}^x\left(q_{k,y}p_k+\frac{1}{q_{k,y}}p_y\right)}.
    \end{align}
    Here, we define \(\rho(t)\), \(\sigma(t)\) (\(t\in[0,x-1]\)) as
    \begin{align}
        \rho(t)&=\sum_{l=1}^{t}\frac{2p_l\prod_{k=1}^{l-1}\left(q_{k,x}p_k-\frac{1}{q_{k,x}}p_x\right)}{\prod_{k=1}^l\left(q_{k,x}p_k+\frac{1}{q_{k,x}}p_x\right)} \\
        &\times\frac{\prod_{k=1}^{l-1}\left(q_{k,y}p_k-\frac{1}{q_{k,y}}p_y\right)}{\prod_{k=1}^l\left(q_{k,y}p_k+\frac{1}{q_{k,y}}p_y\right)}, \\
        \sigma(t)&=\frac{\prod_{k=1}^{t}\left(q_{k,x}p_k-\frac{1}{q_{k,x}}p_x\right)}{\prod_{k=1}^{t}\left(q_{k,x}p_k+\frac{1}{q_{k,x}}p_x\right)} \\
        &\times\frac{\prod_{k=1}^{t}\left(q_{k,y}p_k-\frac{1}{q_{k,y}}p_y\right)}{\prod_{k=1}^{t}\left(q_{k,y}p_k+\frac{1}{q_{k,y}}p_y\right)}.
    \end{align}
    Then we have 
    \begin{align}
        \left(LL^\top\right)_{x,y} = \rho(x-1)+\frac{\sigma(x-1)}{q_{x,y}p_x+\frac{1}{q_{x,y}}p_y}.
    \end{align}
    When \(t\geq 1\),
    \begin{align}
        &\rho(t)+\frac{\sigma(t)}{q_{x,y}p_x+\frac{1}{q_{x,y}}p_y} \\
        &=\rho(t-1)+\frac{2p_t\prod_{k=1}^{t-1}\left(q_{k,x}p_k-\frac{1}{q_{k,x}}p_x\right)}{\prod_{k=1}^t\left(q_{k,x}p_k+\frac{1}{q_{k,x}}p_x\right)} \\
        &\hspace{2em}\times\frac{\prod_{k=1}^{t-1}\left(q_{k,y}p_k-\frac{1}{q_{k,y}}p_y\right)}{\prod_{k=1}^t\left(q_{k,y}p_k+\frac{1}{q_{k,y}}p_y\right)}+\frac{\sigma(t)}{q_{x,y}p_x+\frac{1}{q_{x,y}}p_y} \\
        &=\rho(t-1)+\left(2\left(q_{x,y}p_x+\frac{1}{q_{x,y}}p_y\right)p_t+q_{t,x}q_{t,y}p_t^2 \right.\\
        &\hspace{2em}-\left.\left(\frac{q_{t,y}}{q_{t,x}}p_x+\frac{q_{t,x}}{q_{t,y}}p_y\right)p_t+\frac{1}{q_{t,x}q_{t,y}}p_x p_y\right) \\
        &\hspace{2em}\times\frac{1}{\left(q_{t,x}p_t+\frac{1}{q_{t,x}}p_x\right)\left(q_{t,y}p_t+\frac{1}{q_{t,y}}p_y\right)} \\
        &\hspace{2em}\times\frac{\sigma(t-1)}{\left(q_{x,y}p_x+\frac{1}{q_{x,y}}p_y\right)}.
    \end{align}
    Here,
    \begin{align}
        q_{x,y}=\frac{2^{n_y}}{2^{n_x}}=\frac{2^{n_y}/2^{n_t}}{2^{n_x}/2^{n_t}}=\frac{q_{t,y}}{q_{t,x}}
    \end{align}
    holds, and thus we have 
    \begin{align}
        \rho(t)+\frac{\sigma(t)}{q_{x,y}p_x+\frac{1}{q_{x,y}}p_y} 
        =\rho(t-1)+\frac{\sigma(t-1)}{q_{x,y}p_x+\frac{1}{q_{x,y}}p_y}.
    \end{align}
    As a result, we obtain the following equation:
    \begin{align}
        \left(LL^\top\right)_{x,y} &= \rho(0)+\frac{\sigma(0)}{q_{x,y}p_x+\frac{1}{q_{x,y}}p_y} \\
        &=\frac{1}{q_{x,y}p_x+\frac{1}{q_{x,y}}p_y}=k_{\rm SH}(G_x, G_y)/2. 
    \end{align}
    (iii) The case \(x>y\). 
    The result directly follows by exchanging \(x\) and \(y\) in the discussion given in the case (ii); 
    \begin{align}
        \left(LL^\top\right)_{x,y}
        &=\sum_{l=1}^y L_{x,l}L_{y,l}
          =\frac{1}{q_{y,x}p_y+\frac{1}{q_{y,x}}p_x} \\
        &=\frac{1}{q_{x,y}p_x+\frac{1}{q_{x,y}}p_y}=k_{\rm SH}(G_x, G_y)/2.
    \end{align}
    Summarizing, we have Eq.~\eqref{eq:cholesky}, meaning that the kernel 
    $K_{\rm SH}(G, G')$ is positive semidefinite.
\end{proof}

\subsubsection*{Relationship with the Bhattacharyya kernel}

In this paper, we proposed two kernels: the Bhattacharyya (BH) kernel \eqref{eq:kernel_def_nff} 
arising in the case of swap test and the SH kernel \eqref{eq:kernel_def_SH} arising in the case 
of switch test. 
The following relationship holds:

\begin{lemma} \label{thm:cosinesimilarity-geq-ourkernel}
    The BH kernel \eqref{eq:kernel_def_nff} is bigger than or equal to the SH kernel  
    \eqref{eq:kernel_def_SH}. 
\end{lemma}

\begin{proof}
    Note that $f_{G}^\top f_{G'}$ is common in both kernels. 
    From the inequality of the arithmetic and geometric means, we obtain the following inequality:
    \begin{align}
        K_{\rm SH}(G, G') 
        &=\frac{2}{\frac{2^{n'}}{2^{n}}\sum_{k=1}^a |X_{k}|^2+\frac{2^{n}}{2^{n'}}\sum_{k=1}^a |X_{k}'|^2}
               f_{G}^\top f_{G'} \\
        & \hspace{-4em}
        \leq\frac{1}{\sqrt{\frac{2^{n'}}{2^{n}}\sum_{k=1}^a |X_{k}|^2}
                                    \sqrt{\frac{2^{n}}{2^{n'}}\sum_{k=1}^a |X_{k}'|^2}}
              f_{G}^\top f_{G'} \\
        & \hspace{-4em}
        =\frac{1}{\sqrt{\sum_{k=1}^a |X_{k}|^2}\sqrt{\sum_{k=1}^a |X_{k}'|^2}}
              f_{G}^\top f_{G'} = K_{\rm BH}(G, G').
    \end{align}
\end{proof}

This result means that, as a similarity measure, the SH kernel is more conservative than the BH kernel 
(note that both kernels are upper bounded by 1, because they originate from the inner product 
$\braket{G|G'}$). 
Figure~\ref{fig:cosinesimilarity-ourkernel} depicts the relation between these kernels for the case of 
MUTAG training dataset; this result implies that, in general, the difference between the kernel values 
would be tiny, but as mentioned above, there will be a solid difference in the prediction classification 
accuracy in view of the conservativity of the corresponding classifiers. 
\\

\begin{figurehere}
    \includegraphics[width=0.7\linewidth]{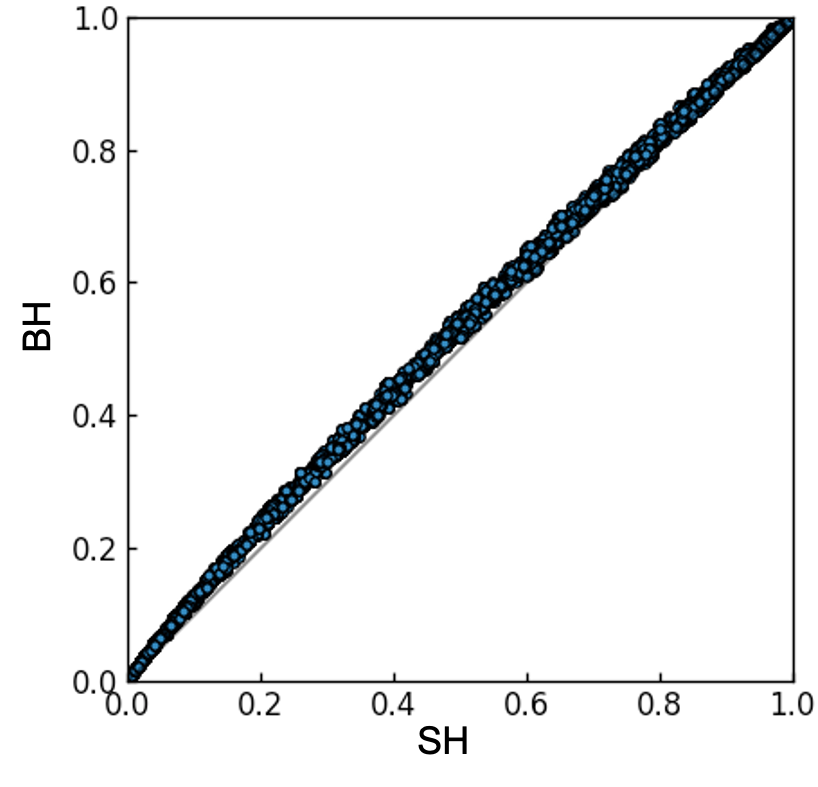}
    \caption{Relationship between BH (vertial axis) and SH (horizontal axis), for the case of 
    MUTAG dataset.}
    \label{fig:cosinesimilarity-ourkernel}
\end{figurehere}
\mbox{}

\subsection*{Time complexity for constructing $E(G,x)$}

\begin{lemma} \label{thm:adder}
    The time complexity of adding 1 to the state $\ket{k} (0\leq k \leq A)$ is $O((\log A)^2)$. 
\end{lemma}

\begin{proof}
We use the multi-controlled NOT gate $C^{i-1} X$, which operates the NOT gate $X$ on the 
\(i^{\text{th}}\) target qubit, if the value of the control qubits (i.e., the $0^{\text{th}}, 1^{\text{st}}, 
\ldots, (i-1)^{\text{th}}$ qubits) are all 1. 
The adder of 1 to $k$ can be done by repeatedly applying $C^{i-1} X$ for $i=1,\ldots,\log k$ 
on the state $\ket{k}$. 
Now $C^{i-1} X$ is composed of $O(i)$ Toffoli gates~\cite{nielsen}; 
in other words, the time complexity for operating $C^{i-1} X$ is $O(i)$. 
Hence, the maximum of the total time complexity of adding 1 to the state $\ket{k}$ is
    \begin{equation}
        \sum_{i=1}^{\log A} O(i) = O((\log A)^2).
    \end{equation}
\end{proof}

\begin{lemma} \label{thm:v}
The time complexity for preparing the feature state $\ket{\#v}$ of a subgraph $x$ is 
$O(|V|(\log |V|)^2)$.
\end{lemma}

\begin{proof}
The oracle is realized as the set of adder of 1 controlled by the index state $\ket{x}$; that is, 
if the $i$th index qubit is $\ket{1}$ (meaning that the $i$th vertex is contained in the subgraph $x$), 
then the corresponding controlled adder adds 1 on the feature state. 
The index runs from $i=0$ to $i=n=|V|$, and hence there are totally $n$ controlled adders. 
Also each controlled adder needs $O((\log n)^2)$ time complexity due to Lemma 3. 
Therefore, the total time complexity is $n\times O((\log n)^2) = O(n(\log n)^2)$.
\end{proof}

\begin{lemma} \label{thm:e}
The time complexity for preparing the feature state $\ket{\#e}$ of a subgraph $x$ is $O(|E|(\log |E|)^2)$. 
\end{lemma}

\begin{proof}
The idea is the same as that of Lemma 4, except that the adder is controlled by the pair of index qubits 
representing an edge of the subgraph $x$. 
Because there are $|E|$ such pairs, the total time complexity is 
$|E|\times O((\log |E|)^2) = O(|E|(\log |E|)^2)$. 
\end{proof}


\begin{lemma} \label{thm:dD}
The time complexity for preparing the feature state $\ket{\#dD}$ (the number of vertices 
with degree $D$, where $D=1, 2, 3$) of a subgraph $x$ is $O(n((\log n)^2+d(\log d)^2))$. 
\end{lemma}

\begin{proof}
Recall that $d$ is the maximum degree of the graph $G$. 
First, we store the degree of the $i$th qubit, $d$, into an auxiliary $\log d$ qubits. 
This can be done by performing the adder controlled by each qubit to which the target vertex 
is connected. 
This operation requires $\log d$ qubits and $O(d(\log d)^2)$ time complexity. 
Next, if the stored value is $D$, we perform the adder controlled by the auxiliary qubits, on 
the feature state. 
The required space of $\#dD$ is $O(\log n)$, because $\#dD\leq n$. 
The time complexity is $O((\log n)^2)$. 
Finally, we initialize the auxiliary qubits by the inverse operation. 
We perform these calculation recursively from the $0^{\text{th}}$ qubit to the $n-1^{\text{th}}$ qubit. 
As a result, the total time complexity is 
\begin{align}
    & n\times(O(d(\log d)^2) + O((\log n)^2) + O(d(\log d)^2)) \\
    &\hspace{2em} = O(n((\log n)^2+d(\log d)^2)).
\end{align}
\end{proof}

\subsection*{Note on the simulation method}

In the numerical experiment, we studied various type of graph set containing a graph with 
28 vertices, in which case we need a quantum device composed of 56 qubits. 
A naive implementation of the numerical simulator would require a $2^{56}$ memory, 
probably larger than 1 Exabit depending on the accuracy, which is not realistic. 
Hence, we adopted a parallel GPU computation; that is, we compute $\ket{E(G,x)}$ for each 
subgraph $x$ and store it in each GPU memory, which thus needs a $2^{28}$ memory.

\section*{Data Availability}
A complete set of the kernel values and the probabilities for removing the index state are available at 
\url{https://github.com/TRSasasusu/GraphKernelEncodingAllSubgraphsQC}.

\section*{Code Availability}
The codes for computing the kernel values and the classification protocols are available at 
\url{https://github.com/TRSasasusu/GraphKernelEncodingAllSubgraphsQC}.

\printbibliography

\section*{Acknowledgements}
This work was supported by MEXT Quantum Leap Flagship Program Grant Number 
JPMXS0118067285 and JPMXS0120319794.

\end{multicols}
\end{document}